\documentclass[12pt]{article}
\usepackage{amsmath}
\usepackage{amsthm}
\usepackage{amssymb}
\usepackage[noend]{algpseudocode}
\usepackage{algorithm}
\usepackage{graphicx}
\usepackage{color}
\usepackage{url}
\usepackage{microtype}
\usepackage{tabu}
\usepackage{lmodern}
\usepackage{fullpage}
\usepackage{sidecap} 
\usepackage[hidelinks]{hyperref} 

\newtheorem{theorem}{Theorem}
\newtheorem{lemma}{Lemma}
\newtheorem{corollary}{Corollary}
\newtheorem{definition}{Definition}
\newtheorem{fact}{Fact}

\newcommand{\opt}{\ensuremath{\text{\footnotesize\textsf{OPT}}}}
\newcommand{\cost}{\ensuremath{\text{\footnotesize\textsf{COST}}}}
\newcommand{\E}{\mathcal{E}}

\title{Online Facility Location on Semi-Random Streams}
\author{Harry Lang\thanks{Inria Saclay and Johns Hopkins University.}}
\date{}

\begin{document}

\maketitle
\thispagestyle{empty}

\begin{abstract}
In the streaming model, the order of the stream can significantly affect the difficulty of a problem.
A $t$-semirandom stream was introduced as an interpolation between random-order ($t=1$) and adversarial-order ($t=n$) streams where an adversary intercepts a random-order stream and can delay up to $t$ elements at a time.
IITK Sublinear Open Problem \#15 asks to find algorithms whose performance degrades smoothly as $t$ increases.
We show that the celebrated online facility location algorithm achieves an expected competitive ratio of $O(\frac{\log t}{\log \log t})$.
We present a matching lower bound that any randomized algorithm has an expected competitive ratio of $\Omega(\frac{\log t}{\log \log t})$.

We use this result to construct an $O(1)$-approximate streaming algorithm for $k$-median clustering that stores $O(k \log t)$ points and has $O(k \log t)$ worst-case update time.
Our technique generalizes to any dissimilarity measure that satisfies a weak triangle inequality, including $k$-means, $M$-estimators, and $\ell_p$ norms.
The special case $t=1$ yields an optimal $O(k)$ space algorithm for random-order streams as well as an optimal $O(nk)$ time algorithm in the RAM model, closing a long line of research on this problem.
\end{abstract}

\newpage
\setcounter{page}{1}

\section{Introduction}

One of the fundamental theoretical questions in the streaming model is to understand how the stream order impacts computation.
In the adversarial-order model, results must hold under any order, whereas in the random-order model the order is selected uniformly at random.
The order of the stream can strongly affect the resources required to solve a problem.
For example, for streams of $n$ integers where the stream may be read sequentially multiple times, determining the median using polylogarithmic space requires $\Theta(\frac{\log n}{\log \log n})$ passes in adversarial-order~\cite{MP78, GM07a} but only $\Theta(\log \log n)$ passes in random-order~\cite{GM06, CJP08}.

As demonstrated by the median problem, there can be an exponential gap in the resources required for random-order and adversarial-order streams.
To interpolate between these two extremes, Guha and McGregor introduced two notions of \textit{semirandom-order} where the adversary has limited power.

\begin{definition}[$t$-semirandom order,~\cite{GM06}] \label{def:semirandom}
A $t$-bounded adversary\footnote{Our definition of $t$-bounded actually corresponds to $(t-1)$-bounded as introduced in~\cite{GM06}.  This turns out to be more natural and avoids writing $t+1$ in all our bounds.} can permute a stream $p_1, \ldots, p_n$ to the stream $p_{\sigma(1)}, \ldots, p_{\sigma(n)}$ with any permutation $\sigma$ that satisfies $|\{j \in [n] : j < i \text{ and } \sigma(j) > \sigma(i)\}| < t$ for every $i \in [n]$.
A stream is in $t$-semirandom order if it is generated by a $t$-bounded adversary acting on a random-order stream.
\end{definition}

\begin{definition}[$\epsilon$-generated random order,~\cite{GM07a}]
Let $\mu$ be the uniform distribution over all permutations of $n$ elements.
A stream of $n$ points arrives in $\epsilon$-generated random-order if the permutation is drawn from a distribution $\nu$ such that $|| \mu - \nu ||_1 \le 2 \epsilon$.
\end{definition}

\noindent
These models capture the notion of an adversary with limited power and establish a spectrum of semirandom orders to intermediate between 
the fully random and fully adversarial cases.
IITK Sublinear Open Problem \#15~\cite{sublinear15} asks:
\begin{quotation}
\noindent
How do these notions relate to each other? Can we develop algorithms whose performance degrades smoothly as the stream ordering becomes ``less-random" using either definition? For a given application, which notion is more appropriate?
\end{quotation}

\noindent
We respond to the first question by showing that no non-trivial relations hold between these models.
One can verify that $\epsilon = 0$ and $t = 0$ correspond to random-order, and that $\epsilon = 1$ and $t = n$ correspond to adversarial-order.
However, we show that these models are incomparable in the sense that an $\epsilon$-generated adversary requires $\epsilon > 1 - 2^{-\Omega(n)}$ to simulate the action of a $t$-bounded adversary for any $t > 1$, and that a $t$-bounded adversary requires $t > n/2$ to simulate the action of an $\epsilon$-generated adversary for any $\epsilon \ge 2^{-\Omega(n)}$.

We answer the second question by proving matching bounds for the online facility location problem that show the performance degrades smoothly as $t$ increases.
These are the first bounds for semirandom streams  
that match at all values of $t$.
Previous results matched only for $t$ sufficiently small.
For example, the result of~\cite{GM06} shows how to return the median in $O(\log \log n)$ passes when $t = O(\sqrt{n})$.
However, 
for a constant $c < 1$ the algorithm is only guaranteed to terminate in $O(n)$ passes when $t = \Omega(n^c)$.
In comparison, even at $t = n$ there are polylogarithmic-space algorithms that return the median in only $O(\log n / \log \log n)$ passes~\cite{MP78}.

Our results provide evidence that $t$-bounded adversarial-order is a viable model of semirandomness.
We address the third question by complementing our positive results for the $t$-semirandom model with an argument showing that the $\epsilon$-generated random order model is uninteresting for a wide class of problems.
A more complete discussion of IITK Open Question \#15 is included in Section~\ref{section:open}.


\subsection{Our Contributions}

We present results for online facility location and a large class of clustering problems.
In Section~\ref{section:ofl}, we provide a novel analysis that the online facility location algorithm of~\cite{M01} is $O(\frac{\log t}{\log \log t})$-competitive in expectation.
Adapting Meyerson's original argument for a $t$-bounded adversary is possible but results in an $O(t)$ expected competitive ratio.
We introduce a different analysis that permits this exponential improvement.
We complement this result by presenting a matching lower bound in Section~\ref{section:lower} that any randomized algorithm for online facility location is $\Omega(\frac{\log t}{\log \log t})$-competitive in expectation.  
See Table~\ref{table:ofl} for a comparison with existing results.

In Section~\ref{section:clustering}, we present a streaming algorithm for clustering using any 
function that satisfies a weak triangle inequality (this includes $k$-median, $k$-means, $M$-estimators, and $\ell_p$ norms).
Our algorithm stores $O(k \log t)$ points and has $O(k \log t)$ worst-case update time. 
As shown in Table~\ref{table:clustering}, we match the state-of-the-art for adversarial-order streams and provide the first results for all $t < n$.
We remark that our algorithm respects sparsity by only storing a weighted subset of the input.
Another notable property of our clustering algorithm is that it is oblivious to the actual values of $k$ and $t$.
The algorithm takes an input value $m$, and the output is valid as long as $m = \Omega(k \log t)$.
This may be useful for practical applications where the number of clusters or power of the adversary is unknown.
For example, if the data exhibits a hierarchical structure than the resolution of the result (measured by $k$) degrades smoothly as the power of the adversary increases.

The special case $t=1$ yields the first result for clustering on random-order streams.
In the \texttt{RAM} model where we can shuffle the input into random order in linear time, this implies an optimal $O(nk)$ time algorithm and closes a long line of research on the problem.

As a blackbox used by our clustering algorithm, we present a method to compress a weighted set of $n$ distinct points to a weighted set of $\frac{n+k}{2}$ distinct points in linear time while incurring less than twice the optimal cost of clustering to $k$ points.
Our algorithm, based on $2$-coloring a nearest neighbor graph, is presented in Section~\ref{section:compress} as it may be of independent interest.


\begin{table}[ht]
\centering
{\tabulinesep=1.2mm
\begin{tabu}{|c||c|c||c|c|}
	\hline
	Regime & Upper Bound & Source & Lower Bound & Source \\ \hline
	$t=1$ & $O(1)$ & \cite{M01} & $\Omega(1)$ & Trivial \\ \hline
	$1 \le t \le n$ & $O\left(\frac{\log t}{\log \log t}\right)$ & $\star\star$ & $\Omega\left(\frac{\log t}{\log \log t }\right)$ & $\star\star$\\ \hline
	$t=n$ & $O\left(\frac{\log n}{\log \log n}\right)$ & \cite{F07} & $\Omega\left(\frac{\log n}{\log \log n}\right)$ & \cite{F07} \\
	\hline
\end{tabu}}
  \caption{Expected competitive ratio for online facility location. Upper bounds apply to the algorithm of~\cite{M01}. 
  Lower bounds apply to any randomized algorithm.}
  \label{table:ofl}
\end{table}

\begin{table}
\centering
{\tabulinesep=1.2mm
\begin{tabu}{|c||c|c|}
	\hline
	Regime & Space & Source \\ \hline
	$t=1$ & $O(k)$ & $\star\star$ \\ \hline
	$1 \le t \le n$ & $O(k \log t)$ & $\star\star$ \\ \hline
	$t = n$ & $O(k \log n)$ & \cite{F07} \\
	\hline
\end{tabu}}
  \caption{Space complexity (measured in weighted points) of algorithms for metric $k$-median and $k$-means in the streaming model.}
  \label{table:clustering}
\end{table}

\subsection{Prior Work}



\paragraph{Random-Order Streams:}
There has been an increasing interest to design algorithms for data streams that arrive in random-order,
and in recent years the model has become quite popular.
Random-order streams have been considered for problems including rank selection~\cite{MP78, GM09, MV12}, 
frequency moments~\cite{AMO08}, 
entropy~\cite{GMV06}, 
submodular maximization~\cite{MZ15}, 
and graph matching~\cite{KMM12, KKS14, EHLMO15}.
Lower bounds that hold even under the assumption of random-order have been developed using multi-party communication complexity~\cite{CKS03, CCM08, CJP08, GH09}.
Semirandom-order streams, in both the $t$-bounded and $\epsilon$-generated models, have been considered for rank selection~\cite{GM06, GM07a}.
The stochastic streaming model, which takes the random-order assumption a step further by assuming that stream elements are independent samples from an unknown distribution, has also attracted attention~\cite{GM07b, W09, CMVW16}.
The stochastic streaming model is strictly easier than the random-order model since any stochastic stream is automatically in random-order.

\paragraph{Online Facility Location:}
The study of online facility location was initiated by Meyerson~\cite{M01}.
He provided a simple randomized algorithm and proved that for random-order streams it is $O(1)$-competitive in expectation.
Later, Fotakis~\cite{F07} showed that for adversarial-order streams any randomized algorithm has an expected competitive ratio of $\Omega(\log n / \log \log n)$ and proved that Meyerson's randomized algorithm achieves this bound; he also presented a novel deterministic algorithm that achieves this bound.
For Euclidean space, a simple and practical deterministic algorithm was provided by~\cite{ABU04}.

\paragraph{Streaming Metric $k$-median and $k$-means Clustering:}
The streaming $k$-median and $k$-means problems have only been considered in the adversarial-order model.
These problems are well-studied; here we mention only the metric space results that achieved an improvement in the space bound over the previous state-of-the-art.
The first streaming solution computed a $2^{O(1/r)}$-approximation for any $r \in (0,1)$ and stored $O(n^r / r)$ points~\cite{Guha2003}.
Later, an algorithm storing only $O(k \log^2 n)$ points was provided~\cite{Charikar2003}.
The current state-of-the-art $O(1)$-approximation stores $O(k \log n)$ points~\cite{BMOR11}.  
A variety of other results are known for Euclidean space.

\begin{table}[ht]
\centering
{\tabulinesep=1.2mm
\begin{tabu}{|c|c|}
    \hline
    Runtime & Source \\ \hline
    $O(n^2 \log n)$ & \cite{JainVaz2001} \\ \hline
     $O(nk \log k)$ & \cite{Indyk99} \\ \hline
     $O(n^2)$ & \cite{Mettu2003}  \\ \hline
     $O(nk + n\log n + k^2 \log^2 n)$ & \cite{Mettu2004} \\ \hline
     $O(nk + n^{1/2} k^{3/2} \log^2 n \log^{3/2} k)$ & \cite{Chen2009} \\ \hline
     $O(nk)$ & $\star\star$ \\ \hline 
     $\Omega(nk)$ & \cite{Mettu2004} \\
    \hline
  \end{tabu}}
  \caption{Results for metric $k$-median in the RAM model}
  \label{table1}
\end{table}

\paragraph{RAM-Model Metric $k$-median and $k$-means Clustering:}
The history of fast $O(1)$-approximations\footnote{Most of the results shown in Table~\ref{table1} actually output $O(k)$ centers instead of exactly $k$.
However, we observe that the result of~\cite{Mettu2003} implies that any solution of $O(k)$ centers can be converted to a solution of exactly $k$ centers in $O(k^2)$ time.} in the RAM model is summarized in Table~\ref{table1}, omitting results that do not improve the runtime for any value of $k$.
These results for $k$-median generalize to $k$-means with a larger constant in the approximation ratio.
We conclude this line of research with an optimal $O(nk)$ time algorithm, matching the $\Omega(nk)$ time lower bound for any randomized algorithm~\cite{Mettu2004}.
We remark that there exists an $O(nk)$ time algorithm for $k$-means in Euclidean space~\cite{ADK2009}, but it relies on the principal axis theorem and therefore does not generalize to $k$-median or to other metric spaces.

\section{Preliminaries} \label{section:preliminaries}

Let $(\mathcal{X},d)$ be a metric space.
In the facility location problem with parameter $f > 0$ (called the facility cost), we are given a set\footnote{We use the word ``set'' to actually mean ``multiset''.  Multisets may contain multiple copies of the same element.} $A \subset \mathcal{X}$ called \textit{demands}.  The problem is to compute a set $B \subset \mathcal{X}$ called \textit{facilities} and to connect
each demand to a facility.
To connect demand $a$ to facility $b$, we incur cost $d(a,b)$.
We also incur cost $f$ for each facility opened.
The objective is to compute $B$ such that the total cost is minimized.
Defining $\cost(A,B) = \sum_{a \in A} \min_{b \in B} d(a,b)$, the total cost is $|B|f + \cost(A,B)$ by connecting each demand to the nearest facility.


In online facility location, we receive $A$ as a stream of points.
When point $p$ arrives, we may open a facility (incurring facility cost $f$) and then must connect $p$ to a facility (incurring connection cost).
Observe that if we open a facility at the location of $p$, there is no connection cost.
The problem is online because the decisions to open a facility and connect $p$ are irrevocable, meaning that a facility can never be closed and that $p$ cannot be reconnected if a closer facility opens later.

For the $k$-median problem, there is no facility cost but the number of facilities (here called \textit{centers}) is fixed at $k$.
The goal is to compute a set $B$ that minimizes the total cost $\cost(A,B)$.
When the input arrives as a stream, we seek to design algorithms that require a minimal amount of memory.






\begin{definition}[Optimal Cost] \label{def:optKmedian}
Let $A$ be a set and let $k \ge 1$.
$\opt_k(A)$ is defined as the minimum of $\cost(A,B)$ where $B \subset \mathcal{X}$ ranges over all sets of $k$ points.
\end{definition}

The optimal cost for $k$-median is $\opt_k(A)$.
For facility location, the optimal cost is the minimum $kf + \opt_k(A)$ where $k$ ranges over all positive integers.
An $\alpha$-approximation is a solution with cost at most $\alpha$ times the optimum.

A $t$-semirandom stream is the result of a random-order stream that has been intercepted by a $t$-bounded adversary (see Definition~\ref{def:semirandom}). 
Imagine that the stream of elements is a deck of cards,
initially shuffled into random order.
The adversary draws cards into his hand from the deck.
He may give any card from his hand to the algorithm.
The restriction is that he can have at most $t$ cards in his hand at any time.
This means that if he has a full hand of $t$ cards, he cannot draw a new card until giving one to the algorithm.
See Figure~\ref{fig:adversary} for an example.

\begin{SCfigure} \label{fig:adversary}
	\centering
	\includegraphics[scale=0.2]{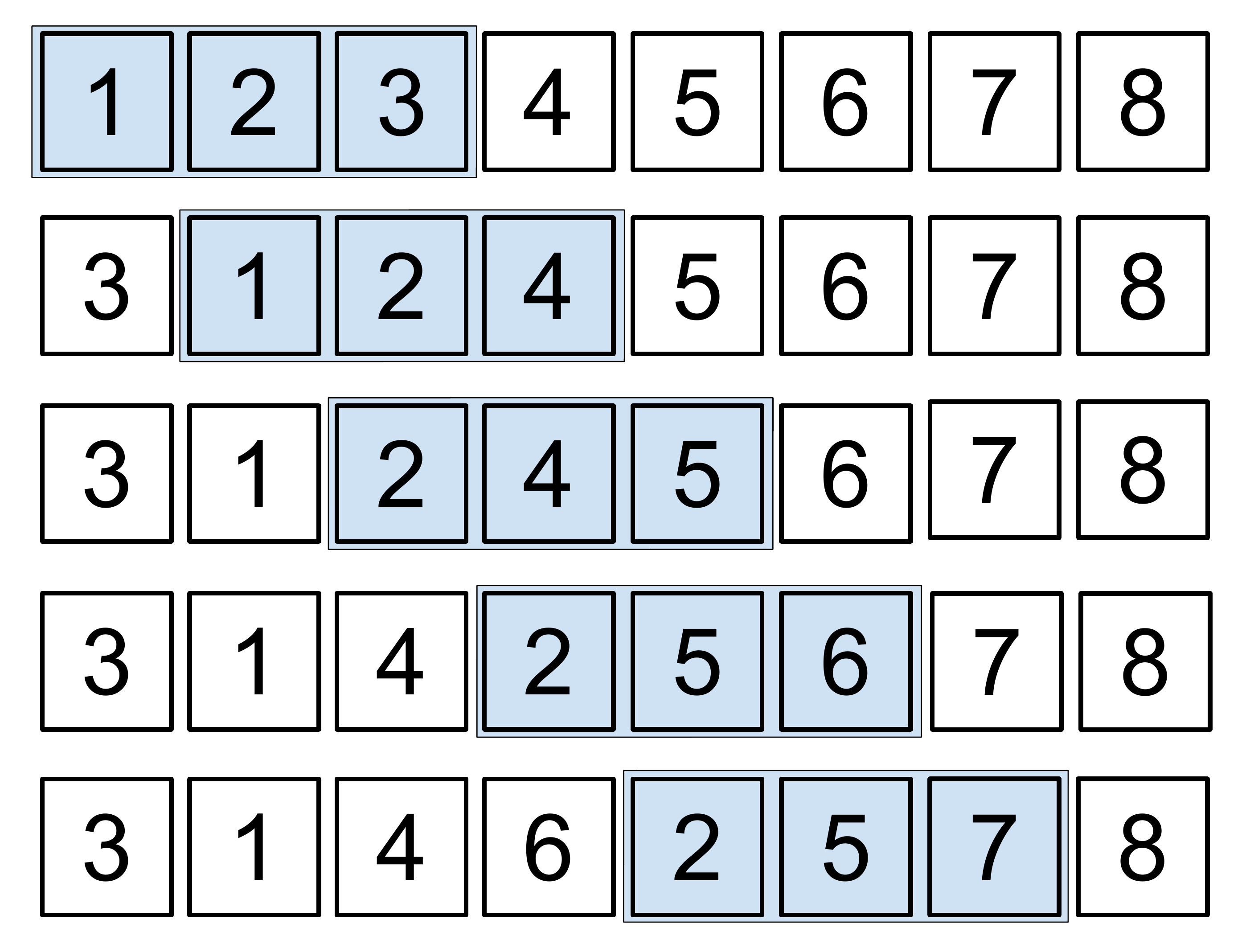}
	\caption{Example of a $3$-bounded adversary acting on a stream.
		Each step is shown on a different row.
		The shaded box represents the memory of the adversary.
		The adversary must send an element to the algorithm before receiving the next element.}
\end{SCfigure}

\section{Online Facility Location} \label{section:ofl}

The algorithm of Meyerson~\cite{M01} is simple and elegant.
Let $f > 0$ be the facility cost parameter.
When a point $p$ arrives, let $\delta(p)$ be the distance between $p$ and the nearest facility.
With probability $\min(1,\delta(p)/f)$, we open a facility at $p$ and pay facility cost $f$.
Otherwise, we connect $p$ to the nearest facility and pay connection cost $\delta(p)$.  We write \texttt{OFL} to refer to this algorithm.

\begin{theorem} \label{thm:general}
Let $r > 1$ and $h \in \mathbb{N}$ such that $r^h \ge 4t$, and let $k$ be a positive integer.
If $\texttt{OFL}(f)$ runs on a $t$-semirandom stream $S$,
the facility cost and connection cost incurred by $\texttt{OFL}(f)$ are each less than $(r+3)\opt_k(S) + (h+2)fk$ in expectation.
\end{theorem}
\begin{proof}
Partition $S$ into $k$ optimal clusters $\{X_\ell\}_{\ell \in [k]}$ so that $\opt_k(S) = \sum_{\ell \in [k]} \opt_1(X_\ell)$.
Consider the centers $\{c_\ell\}_{\ell \in [k]}$ such that $\cost(X_\ell,c_\ell) = \opt_1(X_\ell)$.
Lemmas~\ref{lemma:boundG},~\ref{lemma:boundM}, and~\ref{lemma:boundB} bound the expected connection cost and expected facility cost on each $X_\ell$ by $(r+3)\opt_1(X_\ell) + (h+2)f$.
We obtain the result by summing over all $\ell \in [k]$.
\end{proof}

We now provide the results used in the proof of Theorem~\ref{thm:general}.
Let $X \subset S$ be a set of $n$ points.  Given a center point $c$, define $A = \cost(X,c)$ and $a = A/n$.
We partition $X$ into three pieces as follows:

$G = \{ x \in X : d(x,c) \le 2a \}$
	
$M = \{ x \in X : 2 a < d(x,c) \le 2 r^h a \}$
	
$B = \{ x \in X : d(x,c) > 2 r^h a \}$

\noindent
Defining $A_Z = \cost(Z,c)$ for any set $Z$, we decompose $A = A_G + A_M + A_B$.

\begin{fact} \label{fact:waitTime}
In any order, the expected connection cost of a set $Z$ incurred before a facility opens in $Z$ is less than $f$.
\end{fact}
\begin{proof}
	For convenience we normalize to $f = 1$.
	Let $(z_1, \ldots, z_m)$ be the order of $Z$.
	Define $x_j = \min(1, \delta(z_j))$.
	There is a probability of $x_j$ that point $z_j$ opens as a facility; otherwise $z_j$ incurs connection cost $x_j$.
	
	Let $E_i$ denote the expected connection cost before a facility is opened when \texttt{OFL} is run on the suffix $(z_i, z_{i+1}, \ldots, z_m)$.
	For $1 \le i \le j \le m$ define $P_{i}^{j} = \prod_{\ell=i}^{j} (1 - x_\ell)$. 
	Observe that $E_i = \sum_{j=i}^{m} x_j P_{i}^{j}$.
	We seek to prove that $E_1 < 1$.
	
	We write the recursive formula $E_{i} = (1 - x_{i})(x_i + E_{i+1})$.
	Observe that $E_m \le 1/4$, with the maximum occurring when $x_m = 1/2$.
	Assuming inductively that $E_{i+1} < 1$, observe that $E_i = (1 - x_i)(x_i + E_{i+1}) < 1 - x_i^2 \le 1$.
	We conclude that $E_1 < 1$.
\end{proof}

Observe that $\delta(p)$ can be used to simultaneously bound both the expected connection cost and the expected facility cost.
We use this in Lemmas~\ref{lemma:boundG}-\ref{lemma:boundB} to bound both types of cost with the same argument.

\begin{lemma} \label{lemma:boundG}
In any order, the expected connection cost and expected facility cost of $G$ are each at most $f + A_G + 2a|G|$.
\end{lemma}
\begin{proof}
By Fact~\ref{fact:waitTime}, the expected connection cost of points in $G$ before a facility opens in $G$ is less than $f$.
The facility cost is exactly $f$ when the first facility opens in $G$.
After a facility has opened in $G$, we may bound $\delta(g) \le d(g,c) + 2a$ for any $g \in G$ by the triangle inequality.
The result follows by summing over all $g \in G$.
\end{proof}

The proof of the next lemma is similar to the previous.
\begin{lemma} \label{lemma:boundM}
In any order, the expected connection cost and expected facility cost of $M$ are each at most $hf + (r+1)A_M$.
\end{lemma}
\begin{proof}
We partition $M$ into $h$ parts, defining $M_j = \{ x \in X : 2 r^{j-1} a < d(x,c) \le 2 r^j a \}$ for $j \in \{1, \ldots, h\}$.
By Fact~\ref{fact:waitTime}, the expected connection cost of points in $M_j$ before a facility opens in $M_j$ is less than $f$.
The facility cost of $M_j$ is exactly $f$ when the first facility opens in $M_j$.
After a facility opens in $M_j$, we may bound $\delta(m) \le d(m,c) + 2 r^j a \le (r+1)d(m,c)$ for each $m \in M_j$ by the triangle inequality.
The result follows by summing over all $m \in M$.
\end{proof}

Fixing an order, let $\alpha'_i$ be the expected connection cost of the $i^\text{th}$ point of $G$.  For $i \in \{1, \ldots, |G|\}$, define $\alpha_i = \min_{1 \le j \le i} \alpha'_j$.
For $i \in \{0, \ldots, |G|\}$, let $\beta_i$ be the number of points in $B$ that arrive after exactly $i$ points of $G$ have arrived.

\begin{lemma} \label{lemma:prelimB}
In any order, the expected connection cost and expected facility cost of $B$ are each at most $A_B + 2a|B| + f \beta_0 + \sum_{i=1}^{|G|} \alpha_i \beta_i$.
\end{lemma}
\begin{proof}
The first $\beta_0$ points of $B$ incur cost at most $f \beta_0$ trivially.
If a point $b$ arrives after a point $g$, then $\delta(b) \le d(b,c) + 2a + \delta(g)$ by the triangle inequality.
Hence for a point $b$ that arrives after $i$ points of $G$,
taking expectations and then the minimum over the $i$ preceding points of $G$ shows that $E[\delta(b)] \le d(b,c) + 2a + \alpha_i$.
The result follows by summing over all $b \in B$.
\end{proof}

The remainder of the proof will depend on the assumption of $t$-semirandom order.
For precision, we continue to use $E[\cdot]$ for expectation over the randomness used by $\texttt{OFL}$ and introduce $\E[\cdot]$ for expectation over the randomness of the stream order.
Let $\beta'_i$ be the number of points in $B$ that occur after exactly $i$ points of $G$ in the initial random-order stream before being intercepted by the adversary.
Since the adversary may hold at most $t$ points, if the adversary has received $t + i$ points of $G$ then the algorithm has received at least $i$ points of $G$.
This provides the relation $\beta_0 + \ldots + \beta_i \le \beta'_0 + \ldots + \beta'_{i + t - 1}$ for every $i \ge 0$.
We can view $\beta'_i$ as the number of balls in the $i^\text{th}$ bin when we randomly drop $|B|$ balls into $|G|+1$ bins.

Observe that $f > \alpha_1 \ge \alpha_2 \ge \ldots \ge \alpha_{|G|}$.
The adversary's optimal strategy against our bound in Lemma~\ref{lemma:prelimB} is to delay points in $G$ as long as possible.
We therefore identify the worst-case bound $\beta_0 = \beta'_0 + \ldots + \beta'_{t - 1}$ and $\beta_i = \beta'_{i + t - 1}$ for every $i \ge 1$.
We rewrite $f \beta_0 + \sum_{i=1}^{|G|} \alpha_i \beta_i = f \sum_{j = 0}^{t-1} \beta'_j + \sum_{i=1}^{|G|} \alpha_i \beta'_{i+t-1}$.

The difficulty is that $\alpha_i$ and $\beta_i$ are dependent random variables.
Although they cannot affect each other directly, both $\alpha_i$ and $\beta_i$ depend on the prefix of the stream ending on the $i^\text{th}$ point of $G$.
To overcome this, we split the sum into two pieces and for each piece we find an upper bound on either $\alpha_i$ or $\beta_i$ that holds independently of the prefix.

The next lemma quantifies the intuition that if many points of $G$ have arrived then there must be a facility very close to $G$.
This bound suffices after a constant fraction of $G$ has arrived.

\begin{lemma} \label{lemma:highsum}
In $t$-semirandom order, $\sum_{i = \lceil |G|/2 \rceil + 1}^{|G|} \E[\alpha_i \beta_i] < f \ln(2) / 4t$.
\end{lemma}
\begin{proof}
Let $\mu(g)$ be the indicator random variable that no facility is open in $G$ after $g$ is processed.
Let $G'$ be the set of $i$ points of $G$ that have arrived.
Fact~\ref{fact:waitTime} implies that $\sum_{g \in G} E[\delta(g) \mu(g)] < f$.
Thus there must be some $g \in G'$ such that $E[\delta(g) \mu(g)] < f/i$.
Since $\alpha_i = \min_{g \in G'} E[\delta(g) \mu(g)]$, we bound $\alpha_i < f/i$.

By Markov's inequality $|B| < \frac{n}{2r^h} \le \frac{n}{8t}$ and $|G| \ge \lceil \frac{n}{2} \rceil$.
Then $\E[\beta_i] = \E[\beta'_{i+t-1}] = \frac{|B|}{|G| + 1} < \frac{1}{4t}$ which implies $\E[\alpha_i \beta_i] < \E[(f/i) \beta_i] < \frac{f}{4ti}$.
The result follows since $\sum_{i = \lceil |G|/2 \rceil + 1}^{|G|} \frac{1}{i} < \ln(2)$.
\end{proof}

If we drop $|B|$ balls into $|G| + 1$ bins and condition on the number of balls in $i$ of the bins, the expected number of balls in any other bin is at most that of dropping $|B|$ balls into $|G| - i + 1$ bins.
This bound blows up towards the end of the stream but suffices for the first half.

\begin{lemma} \label{lemma:lowsum}
In $t$-semirandom order, $\sum_{i=1}^{\lceil |G|/2 \rceil} \E[\alpha_i \beta_i] < f \ln(3) / 2t$
\end{lemma}
\begin{proof}
Let $P(i)$ denote the order of the prefix of the stream ending on the $i^\text{th}$ point of $G$.
Observe that $\E[\beta'_i | P(i)] \le \frac{|B|}{|G|-i+1}$ with the maximum occurring when $\beta'_j = 0$ for all $j \in \{0,\ldots,i-1\}$.
This implies $h(i) := \max_{P(i)} \E[\beta_i | P(i)] \le \frac{|B|}{|G|-t-i+2}$.

By causality, points that arrive after the $i^\text{th}$ point of $G$ do not affect $\alpha_i$.
Therefore $\alpha_i$ is a constant when $P(i)$ is fixed.
We can now separate 
\begin{align*}
	\E[\alpha_i \beta_i] & = \E_{P(i)}[\E[ \alpha_i \beta_i | P(i)]] \\
	& = \E_{P(i)}[\E[\alpha_i | P(i)] \cdot \E[\beta_i | P(i)] \\
	& \le \E_{P(i)}[\E[\alpha_i | P(i)] \cdot h(i)] \\
	& = h(i) \E_{P(i)}[\E[\alpha_i | P(i)]] \\
	& = h(i) \E[\alpha_i]
\end{align*}

We have the restraint $\sum_{i=1}^{|G|} \E[\alpha_i] < f$ by Fact~\ref{fact:waitTime}.
Since $h(i)$ is increasing, the sum $\sum_{i=1}^{\lceil |G|/2 \rceil} h(i) \E[\alpha_i]$ is maximized when $\alpha_i = \frac{f}{\lceil |G|/2 \rceil}$ for $i \le \lceil |G|/2 \rceil$.
We bound $\sum_{i=1}^{\lceil |G|/2 \rceil} \E[\alpha_i \beta_i] \le \sum_{i=1}^{\lceil |G|/2 \rceil} \frac{f}{\lceil |G|/2 \rceil}\frac{|B|}{|G|-t-i+2}$.
By Markov's inequality $\frac{|B|}{\lceil |G|/2 \rceil} < \frac{1}{2t}$.
We may assume that $t < \frac{n}{8}$ since otherwise $B = \varnothing$ and there is nothing to show.
We conclude by $\sum_{i=1}^{\lceil |G|/2 \rceil} \frac{1}{|G|-t-i+2} 
= \sum_{y = \lfloor |G|/2 \rfloor - t + 2}^{|G|-t+1} \frac{1}{y} 
< \ln \left( \frac{|G|-t+1}{\lfloor |G|/2 \rfloor - t + 1} \right)
< \ln(3)$.
\end{proof}

We can now provide a bound for $B$ that holds for $t$-semirandom order streams.

\begin{lemma} \label{lemma:boundB}
In $t$-semirandom order, the expected connection cost and expected facility cost of $B$ are each less than $A_B + 2a|B| + f$.
\end{lemma}
\begin{proof}
We substitute the bounds of Lemmas~\ref{lemma:highsum} and~\ref{lemma:lowsum} into Lemma~\ref{lemma:prelimB}.
The last part is to bound for the worst-case adversary $\E[\beta_0] = \sum_{k = 0}^{t-1} \E[\beta'_j] \le \frac{1}{4}$ and observe that $\ln(18) < 3$.
\end{proof}

\subsection{Application to Online Facility Location}

Our main result for online facility location is now a simple corollary of Theorem~\ref{thm:general}.
As shown in the Section~\ref{section:clustering}, Theorem~\ref{thm:general} yields results for both online facility location and $k$-median clustering by applying the theorem with different choices of $r$ and $h$.

\begin{corollary}
	On a $t$-semirandom order stream, \texttt{OFL} is $(2+o(1))\left( \frac{\log_2 t}{\log_2 \log_2 t} \right)$-competitive in expectation.
\end{corollary}
\begin{proof}
	Let $C$ be an optimal facility set for $S$.
	Define $k = |C|$ and observe that $\opt_k(S)$ is the connection cost associated with the optimal solution with facility set $C$.
	The optimal cost for the facility location problem is then $kf + \opt_k(S)$.
	For any $\epsilon > 0$, set $r = (1 + \epsilon) \frac{\log_2 t}{\log_2 \log_2 t}$ and $h = \lceil r \rceil$.
	Observe that $r^h \ge 4t$ for all $t$ greater than some function of $\epsilon$.
	Applying Theorem~\ref{thm:general} to $S$ with these values of $r$, $h$, and $k$ shows that the total expected cost is at most $2(r+3)(kf + \opt_k(S))$.
	This implies the result since $3 < \epsilon r$ by taking $t$ sufficiently large.
\end{proof}

\paragraph{Remark on Aspect Ratio:}
Suppose that we are working in a metric space of aspect ratio\footnote{The aspect ratio of a metric space is the ratio between the maximum and minimum non-zero distance between points.} $\Delta$.
Setting $r^h \ge \Delta$ instead of $r^h \ge 4t$, observe in the proof of Theorem~\ref{thm:general} that $B = \varnothing$.
This yields a bound of $(r + 3) \opt_k(S) + (h+1)kf$ on both the expected facility cost and expected connection cost.
We may therefore replace $t$ with $\min(t,\Delta)$ in our results.
This justifies the lower bound in the following section being constructed in a metric space of aspect ratio $t$.
Our upper and lower bounds match for all choices of $t$ and $\Delta$ by substituting $\min(t,\Delta)$ for $t$.


\section{Lower Bound} \label{section:lower}

We present a lower bound on the expected competitive ratio of any randomized algorithm for the online facility location problem.
The bound holds even when the algorithm can open a facility at any location in the metric space.
The proof works by constructing a sequence of points that converge to the location of an optimal facility.
At each step, there are enough possible locations of the optimal facility that no algorithm can guess (except with negligible probability) the correct location until it is too late.

\begin{theorem}
On a $t$-semirandom order stream for $t \ge 4$, every randomized algorithm for online facility location has an expected competitive ratio of at least $\frac{1}{3} \lceil \frac{\log_2 t}{\log_2 \log_2 t} \rceil$.
\end{theorem}

We will construct a family of inputs and show that any deterministic algorithm has at least a certain competitive ratio when run on an input selected uniformly at random from this family.  The result immediately extends to randomized algorithms by Yao's principle.

\paragraph{The Metric Space:}
Define $m := \lceil \frac{\log_2 t}{\log_2 \log_2 t} \rceil$,
$h := m-1$,
and $D := f/h$.
Let $z$ be a positive integer.
The points of the metric space are the nodes of a complete $z$-ary tree of depth $h$.
This means that the root is at depth $0$ and leaves are at depth $h$.
The distance between a node at depth $i$ and any of its $z$ children is $Dm^{-i} - Dm^{-i-1}$.  The distance between other nodes is obtained by summing the distances of the shortest path between them.

\paragraph{The Family of Inputs:}
The family of inputs is enumerated by the $h^z$ possible strings of $h$ numbers in $\{1, \ldots, z\}$.
We now describe how to construct the input associated with $(b_1, \ldots, b_h) \in \{1, \ldots, z\}^h$.
Define $x_0$ to be the root.
Recursively define $x_i$ to be the $b_i^\text{th}$ child of $x_{i-1}$.
For $i \in \{1, \ldots, h\}$, place $m^i$ points at node $x_i$.
We let the input size be $n$ for any $n \ge t$. 
Since $\sum_{i=1}^h m^i < t \le n$, there will be some remaining points. 
Place all remaining points at the root.
The randomized input is to select a member from this family of $z^h$ inputs uniformly at random.

\paragraph{The Optimal Cost:}
The optimal cost of the input associated with the string $(b_1, \ldots, b_h)$ is at most the cost of the solution that places facilities at the root and at $x_r$, connecting all non-root points to $x_r$.
The distance between $x_i$ and $x_r$ is $\sum_{j=i}^{h-1} (Dm^{-j} - Dm^{-j-1}) = Dm^{-i} - Dm^{-h} < Dm^{-i}$.
Therefore the optimal cost is less than $2f + \sum_{i=1}^h m^i(Dm^{-i}) = 3f$ for every input the the family.

\paragraph{The $t$-Bounded Adversary's Strategy:}
There are $\sum_{i=1}^h m^i < t$ points not located at the root.
This implies that regardless of the order that the points are sent, a $t$-bounded adversary can deterministically ensure that points arrive in non-decreasing order of depth.
The adversary does this by simply sending along any elements at the root while storing the at most $t$ other elements until they are ready to be sent in non-decreasing order of depth.
We consider this arrival order.

\paragraph{The Algorithm's Optimal Strategy:}
We define a cost scheme for the algorithm that is strictly less than the original cost (therefore any lower bound in this easier scheme is valid for the original problem).  This greatly simplifies the analysis by allowing us to isolate an optimal strategy.

Suppose that if a facility is open at certain node, then the connection cost of a point at any ancestor\footnote{If node $a$ is contained in the subtree of node $b$, we say that $a$ is a descendant of $b$ and that $b$ is an ancestor of $a$.} node is zero.  With this modification, we can define an optimal strategy when a point arrives.  If there is an open facility at any descendant node, then connect this point with zero cost.  Otherwise, open a facility with cost $f$ at any descendant leaf node and then connect with zero cost.

If there is no open facility at a descendant node, the second option is optimal since the nearest facility cannot be closer than the parent node.  The parent node is at distance $Dm^{-(i-1)}-Dm^{-i} = f m^{-i}$.  The algorithm, aware of the family of inputs, knows that a total of $m^i$ points are coming at this node.  This means that the total connection cost from this node would be at least $f$.  Therefore it is optimal to pay $f$ and open a new facility.

Given that we will open a new facility, placing it at a descendant leaf node is optimal because it minimizes the connection cost of future points in the stream.  Without loss of generality, we have our deterministic algorithm always open at the descendant leaf node obtained by moving down the first child of each node in the path from the current node.

\paragraph{The Algorithm's Expected Cost:}
Using the algorithm's strategy defined above, one can see that the algorithm incurs zero connection cost.
For the input associated with $(b_1, \ldots, b_h)$, the number of facilities besides the root is just the number of $b_i$ not equal to $1$.
The probability that $b_i \neq 1$ is $1 - \frac{1}{z}$.
The expected total cost is then $(1-\frac{1}{z})h + 1$.
Using sufficiently large $z$ shows that it is not possible to achieve expected cost below $h + 1$.

\paragraph{}
The competitive ratio is therefore at least $(h+1) / 3 = \frac{1}{3} \lceil \frac{\log_2 t}{\log_2 \log_2 t} \rceil$ as desired.  The result extends to randomized algorithms by Yao's principle.

\section{$k$-Median Clustering on Streams} \label{section:clustering}

In this section, we present a $O(1)$-approximation streaming algorithm for $k$-median clustering that stores $O(k \log t)$ points and has $O(k \log t)$ worst-case update time.
The extension to other functions is sketched in Section~\ref{section:M}.
Our algorithm is based on the doubling algorithm of~\cite{Charikar2003}.
Among our innovations is a subroutine $\mathcal{B}(X,k)$ that permits us to improve the update time and approximation ratio of the original algorithm.
$\mathcal{B}(X,k)$ is described in Section~\ref{section:compress} and comes with the following guarantee:

\begin{theorem} \label{thm:compress}
For a weighted set $X$ of $n$ distinct points, suppose that the nearest neighbor in $X$ has already been computed for each $x \in X$.
Given an integer $k \ge 1$, the algorithm $\mathcal{B}(X,k)$ terminates in $O(n)$ time and outputs a pair $(Z,\lambda)$ such that $Z$ is a weighted set of at most $\lfloor \frac{n + k}{2} \rfloor$ points and $\lambda = \cost(X,Z) < 2 \opt_k(X)$.
\end{theorem}

Throughout this section, as in Theorem~\ref{thm:compress}, we overload notation by writing $\cost(A,B)$ where $B$ is a weighted set (with total weight equal to that of $A$).
Recall that for an unweighted set $B$, the function $\cost(A,B)$ denotes the minimum connection cost of connecting the demand set $A$ to the facility set $B$ where each facility can service an unlimited number of demands.
When $B$ is a weighted set, we let $\cost(A,B)$ denote the minimum connection cost under the restraint that each facility must service exactly its weight in demands.

We now present Algorithm~\ref{alg:kmed} to maintain a set $\Psi$ which we show can determine a $O(1)$-approximation to the $k$-median clustering of the stream.
Observe that the main loop of Lines~\ref{line:beginMain}-\ref{line:endMain} always begins with $|\Psi| \le 29m$, which ensures by Theorem~\ref{thm:compress} that the while-loop of  Lines~\ref{line:terminate}-\ref{line:endWhile} always begins with $|\Psi| \le 15m$.

\begin{algorithm}
\caption{Input: integer $m \ge 1$, a stream of points $P$}\label{alg:kmed}

\begin{algorithmic}[1]
	\State $L \gets 0$
	\State $\Psi \gets $ the first $29m$ points of $P$ \label{line:initInsert}
	\Loop \label{line:beginMain}
		\State $(\Psi, \lambda) \gets \mathcal{B}(\Psi,m)$ \label{line:RAM}
		\State $L \gets \max(10 L,\lambda / 3)$ \label{line:L}
		\State $\cost \gets 0$
		\While{$|\Psi| < 29  m$ and $\cost < 14L$} \label{line:terminate}
			\State $p \gets$ next point of $P$ \label{line:getPoint}
			\State $y \gets \arg \min_{y' \in \Psi} d(p,y')$ \label{line:NN}
			\State $u \gets$ uniform random in $(0,1)$
			\If{$uL <  m \, d(p,y)$}
				\State $\Psi \gets \Psi \cup \{p\}$
			\Else
				\State $w(y) \gets w(y) + 1$
				\State $\cost \gets \cost + d(p,y)$
			\EndIf
		\EndWhile \label{line:endWhile}
	\EndLoop \label{line:endMain}
\end{algorithmic}
\end{algorithm}


\begin{lemma} \label{lemma:L}
Except between Lines~\ref{line:RAM} and~\ref{line:L},
$\cost(P,\Psi) < 20L$.
\end{lemma}
\begin{proof}
The variable $\cost$ is an upper bound on the increase of $\cost(P,\Psi)$ during the current instance of the while-loop.
Since $\cost$ increases by at most $\frac{L}{m}$ in each iteration, 
the termination condition of Line~\ref{line:terminate} ensures that $\cost < 15L$.
It remains to show that the cost when the while-loop began was at most $5L$.
Let us recursively assume that when the previous while-loop began, the cost was at most $5L'$ where $L'$ was the previous value of $L'$.
During that instance, less than $15L'$ cost was incurred.
On Line~\ref{line:RAM}, exactly $\lambda$ cost was incurred.
We may bound $\lambda \le  3 L$ and $L' \le L/10$ by Line~\ref{line:L}.
Therefore the total cost is $5L' + 15L' + \lambda \le 5L$ as desired.
\end{proof}

The next lemma addresses a subtle issue that only arises for $1 < t < n$.  Observe that any segment of a random-order ($t = 1$) stream is in random-order, and that any segment of an adversarial-order ($t=n$) stream is in adversarial-order.
However, a segment of a $t$-semirandom order stream for $1 < t < n$ is not necessarily in $t$-semirandom order (or even in $2t$-semirandom order) because the adversary may have as many as $t$ points in storage when the segment begins.
Instead, we analyze a segment as two separate $t$-semirandom streams, one coming from the adversary's storage and the other as those points that the adversary has not yet received.

\begin{lemma} \label{lemma:opt}
Assume that $m \ge k (4 + \lceil \log_2 t \rceil)$.
Whenever the while-loop of Lines~\ref{line:terminate}-\ref{line:endWhile} terminates, $\opt_k(P) > L$ with probability $\frac{1}{2}$.
\end{lemma}
\begin{proof}
Observe that the while-loop runs \texttt{OFL} with facility cost $f = \frac{L}{m}$.
Setting $r = 2$ and $h = \lceil \log_2 (4t) \rceil$, apply Theorem~\ref{thm:general} and plug in the value for $f$.  This shows that on a $t$-semirandom stream, \texttt{OFL} incurs less than $5 \opt_k(P) + L$ in expected connection cost and opens less than $m(1 + 5 \opt_k(P)/L)$ facilities in expectation.
Observe that this bound holds for the points of $P$ even when $P$ is interlaced with points from another stream.

For a segment $P$ of the $t$-semirandom stream, let $P_1$ be the adversary's storage at the beginning of the segment and let $P_2$ be all other points.
Applying the previous argument twice, we see that on this segment \texttt{OFL} incurs less than $5 \opt_k(P) + 2L$ in expected connection cost and opens less than $m(2 + 5 \opt_k(P)/L)$ facilities in expectation.  The terms involving $\opt_k(P)$ did not double since $\opt_k(P_1) + \opt_k(P_2) \le \opt_k(P)$.
With probability at least $\frac{1}{2}$, the cost and number of facilities are most twice these bounds by Markov's inequality.


Suppose that $\opt_k(P) \le L$.
Then with probability at least $\frac{1}{2}$, the while-loop incurs less than $14L$ cost and opens less than $14m$ facilities.
Since the while-loop begins with $|\Psi| \le 15m$, the termination condition means that either $14L$ cost was incurred or $14m$ facilities were opened.
We conclude with probability at least $\frac{1}{2}$ that $\opt_k(P) > L$ when the while-loop terminates.
\end{proof}

For correctness of the algorithm, the result of the preceding lemma only needs to hold for the most recent termination of the while-loop.  We therefore apply it only once, avoiding a factor of $O(\log n)$ in the space bound that would result from applying the lemma at each loop iteration.

After processing a point, Algorithm~\ref{alg:kmed} waits on Line~\ref{line:getPoint} for the next point.  We now extend the previous lemma to hold on this line, and therefore after each point has been processed.

\begin{lemma} \label{lemma:opt'}
Assume that $m \ge k (4 + \lceil \log_2 t \rceil)$.
On Line~\ref{line:getPoint}, $L < 14 \, \opt_k(P)$ with probability $\frac{1}{2}$.
\end{lemma}
\begin{proof}
Let $L'$ and $\Psi'$ be the states of $L$ and $\Psi$ at the beginning of the current iteration of the main loop.
We condition upon the event that $\opt_k(P) > L'$, which occurs with probability at least $\frac{1}{2}$ by Lemma~\ref{lemma:opt}.
From Line~\ref{line:L} we infer that either $L = 10L'$ or $ 3L = \lambda = \cost(\Psi', \Psi)$.

In the case that $L = 10L'$, then the result is immediate.
Otherwise, $ 3 L = \cost(\Psi', \Psi) < 2 \opt_m(\Psi')$ by the guarantee of $\mathcal{B}(X,k)$.
We have $\opt_m(\Psi') \le \opt_k(\Psi')$ since $m \ge k$ by assumption.
Applying the triangle inequality to each point of $\Psi'$, we get $\opt_k(\Psi') \le \opt_k(P) + \cost(P,\Psi')$.
We have that $\cost(P,\Psi') < 20 L' < 20 \opt_k(P)$, where the first inequality is by Lemma~\ref{lemma:L} and the second inequality is the event we have conditioned upon.
Therefore $L < 14 \opt_k(P)$.
This result holds regardless of how many iterations of the while-loop have occurred since $\opt_k(P)$ is non-decreasing as points are added to $P$.
\end{proof}

We now state our main theorem for clustering.
Amplifying the probability of success is simple; run $\lceil \log_2 \frac{1}{\delta} \rceil$ independent instances of Algorithm~\ref{alg:kmed} in parallel and return the $\Psi$ from an instance with minimal $L$.

\begin{theorem} \label{thm:kmed}
Algorithm~\ref{alg:kmed} with parameter $m$ can be implemented to run in $O(m)$ worst-case update time and store $O(m)$ points.  Suppose that the input stream $P$ of $n$ points arrives in $t$-semirandom order.  If $m \ge k (4 + \lceil \log_2 t \rceil)$, then with probability at least $\frac{1}{2}$ the set $\Psi$ maintained by the algorithm satisfies $\cost(P,\Psi) \le O(1) \cdot \opt_k(P)$.
\end{theorem}
\begin{proof}
Combining Lemmas~\ref{lemma:L} and~\ref{lemma:opt'} shows that 
$\cost(P,\Psi) < 280 \, \opt_k(P)$.
It is immediate from the pseudocode that the storage is less than $29 m$ points.
As for the update time, each iteration of the while-loop requires $O(m)$ time.
By Theorem~\ref{thm:compress}, if the nearest neighbor function for $\Psi$ has been computed then Line~\ref{line:RAM} terminates in $O(m)$ time.
We must show how to ensure with $O(m)$ worst-case update time that the nearest neighbor function for $\Psi$ has been computed before each time that Line~\ref{line:RAM} is executed.

Given the nearest neighbor function $\pi$ for a set of $n$ points, observe that we can insert a point into the set and update $\pi$ in $O(n)$ time.
Beginning with the nearest neighbor function $\pi$ of the first two points of $\Psi$, we simply update $\pi$ with the next three points of $\Psi$ each time one point is received from the stream.
The while-loop runs at least $14m$ times before each time that Line~\ref{line:RAM} executes.
Therefore the nearest neighbor function for all of $\Psi$ is guaranteed to have been computed by the time the while-loop terminates.
\end{proof}

By tweaking parameters and refining the analysis, one can improve the guarantee to $\cost(P,\Psi) < 3 \opt_k(P)$ (in particular, the space blows up as the constant approaches $2$).
It is well-known that if $\cost(P,\Psi) \le \alpha \opt_k(P)$ then any $\gamma$-approximation of $\Psi$ is a $(\alpha(\gamma + 1) + \alpha)$-approximation of $P$~\cite{Charikar2003}.
Therefore Theorem~\ref{thm:kmed} implies that $\Psi$ carries enough information to determine a $O(1)$-approximation of $P$.
Our constant $\alpha \le 3$ does not guarantee a particularly low approximation ratio.
However, Algorithm~\ref{alg:kmed} can be used as a building block for a more accurate solution.
Using the technique of~\cite{BFL16}, we can use Algorithm~\ref{alg:kmed} to maintain an $\epsilon$-coreset which carries enough information to determine a $(1+\epsilon)$-approximation\footnote{An efficiently computed solution will have a larger approximation factor.  Both the $k$-median and $k$-means problem are \texttt{MAX-SNP} Hard for all $k \ge 2$.  See the related work section of~\cite{ACKS15} for a survey of hardness results.} for any $\epsilon > 0$.
This technique essentially converts the constant in the approximation factor into a constant in the size of the coreset.
The only space required in addition to Algorithm~\ref{alg:kmed} is the space needed to store the $\epsilon$-coreset.
As an example, for $k$-median in $\mathbb{R}^d$, coresets of size $O(\epsilon^{-2} kd)$ are known~\cite{FL11}, implying that our result can be used to determine a $(1+\epsilon)$-approximation using $O(\epsilon^{-2} k d + k \log t)$ space.

As a corollary to Theorem~\ref{thm:kmed}, we obtain an $O(nk)$-time approximation algorithm for the RAM model.
In light of the $\Omega(nk)$ time lower bound of~\cite{Mettu2004}, the runtime is optimal.

\begin{corollary}
Given a set $P$ of $n$ points, there exists an algorithm that outputs a $O(1)$-approximation to the $k$-median clustering of $P$ with probability $1-\delta$ in time $O(n k \log \frac{1}{\delta})$.
\end{corollary}
\begin{proof}
Shuffle $P$ in $O(n)$ time.
Set $m = 4k$ and run Algorithm~\ref{alg:kmed} in $O(nk)$ time followed by the offline algorithm of~\cite{Mettu2004} in $O(k^2)$ time to obtain by Theorem~\ref{thm:kmed} a set $C$ of $k$ points such that $\cost(P,C) \le O(1) \cdot \opt_k(P)$ with probability at least $\frac{1}{2}$.
Repeat this $\lceil \log_2 \frac{1}{\delta} \rceil$ times and output the solution of minimal cost.
\end{proof}

\subsection{Extension to Other Functions} \label{section:M}

We have assumed that we are in a metric space $(\mathcal{X},d)$ but we can weaken this assumption.
Suppose that throughout our results we replace the metric with an arbitrary symmetric positive-definite function $D : \mathcal{X} \times \mathcal{X} \rightarrow [0, \infty)$.
If $D$ satisfies the triangle inequality, then $D$ is a metric and our result for $k$-median applies directly.
However, suppose that $D$ just satisfies a weak triangle inequality for some $\beta \ge 1$:
$$D(a,c) \le \beta (D(a,b) + D(b,c)) \text{ for all } a,b,c \in \mathcal{X}$$
All of our proofs go through with larger constants.
The bound in Theorem~\ref{thm:kmed} generalizes to $(h+1+\beta)fk + \beta(r+3\beta)\opt_k(S)$ and the guarantee of the $\mathcal{B}(X,k)$ routine of the next subsection generalizes to $\cost(X,Z) < 2 \beta \opt_k(X)$.
For any function $D$ satisfying a constant $\beta$ value, 
our results for both online facility location and clustering carry through with larger constants.
An example application is that our results generalize to $\ell_p$ norms.
Another important case is $k$-means that corresponds to $D(x,y) = d(x,y)^2$.

Recall that the maximum likelihood estimator for the mean $\mu$ of Gaussian data $Q$ is the $\hat{\mu}$ that minimizes $\sum_{q \in Q} d(q,\hat{\mu})^2$.
To handle outliers more robustly, the statistics community introduced $M$-estimators which generalize maximum likelihood estimation by minimizing $\sum_{q \in Q} \rho(d(q,\hat{\mu}))$ for some positive-definite function $\rho : [0, \infty) \rightarrow [0, \infty)$.
An $M$-estimator along with a positive integer $k$ defines a clustering problem to find a set $C$ of $k$ points that minimizes $\sum_{q \in Q} \min_{c \in C} \rho(d(q,c))$.
Observe that we recover $k$-means for $\rho(x)= x^2$ and $k$-median for $\rho(x) = x$.
The convergence and robustness properties of $M$-estimators have been well-studied, but we also observe that the function $D = \rho \circ d$ usually satisfies a weak triangle inequality for a very low $\beta$ value.
Evidently we can let $\beta$ be any value such that $\rho(c) \le \beta(\rho(a) + \rho(b))$ for all $a,b,c \ge 0$ such that $c \le a + b$.
In Table~\ref{table:rho} we have calculated tight $\beta$ values for the most commonly used $M$-estimators .

\begin{table}[ht]
\centering
{\tabulinesep=1.2mm
\begin{tabu}{|c|c|c|c|}
  \hline
  Estimator & $\rho$ function & $\beta$ \\
  \hline
  Linear & $\rho(x) = x$ & 1 \\
  \hline
  Gaussian & $\rho(x) = x^2$ &  2  \\
  \hline
  Huber & $\rho(x)=
	\begin{cases}
	x^2 &\text{if } x \text{ \textless } \text{ } 1\\
	2x - 1 & \text{if x}  \ge 1
	\end{cases}$ & 2 \\
  \hline
  Cauchy   & $\rho(x)= \log (1 + x^2)$ & 2 \\
  \hline
  Tukey   & $\rho(x)=
	\begin{cases}
		1-(1-x^2)^3 &\text{if } x \text{ \textless } \text{ } 1\\
		1 & \text{if x}  \ge 1
	\end{cases}$ & 2 \\
  \hline
\end{tabu}}
\caption{List of the most common $M$-estimators.
The last column is redundant since the $\beta$ value can be directly calculated from the $\rho$ function.
All of these estimators are actually a parameterized family by scaling $\rho(x)$ to $c_1 \rho(c_2 x)$ for $c_1,c_2 > 0$, but we choose only a single representative since the $\beta$ value is unchanged.}
\label{table:rho}
\end{table}



These results imply that $O(k \log t)$ space suffices to approximate $M$-estimators on a data stream.
We remind the reader that we can use the technique of~\cite{BFL16} to maintain an $\epsilon$-coreset which permits a $(1+\epsilon)$-approximation to the optimal $M$-estimator.


\subsection{The $\mathcal{B}(X,k)$ Routine} \label{section:compress}

We present a deterministic algorithm $\mathcal{B}$ that accepts a weighted set $A$ of $n$ distinct points along with a positive integer $k$ and returns a weighted set $Z$ of $\lfloor \frac{n + k}{2} \rfloor$ distinct points such that $\cost(A,Z) < 2 \opt_k(A)$.
If the nearest neighbor graph on $A$ has been computed, then $\mathcal{B}$ terminates in $O(n)$ time.

In what follows, there must be a way to order the points of $A$.  This is necessary for a technical detail that comes up in Lemma~\ref{lemma:cycles}; we need a consistent way to break ties.  In practice, this can simply be the order that the algorithm loops through the points of $A$.

\begin{definition} \label{definition:pi}
	The function $\pi : A \rightarrow A$ is defined such that $\pi(a) = \arg \min_{x \in A \setminus \{a\}} d(a,x)$.
	If there is more than one point $x$ that minimizes $d(x,a)$, break ties by selecting the $x$ that is greatest according to the order of $A$.
\end{definition}

We use $\pi$ to define a graph as follows:

\begin{definition}
	Let $A$ be a weighted set.
	The directed graph $G(A,\pi)$ has the points of $A$ as vertices.
	For each vertex $a$, there is exactly one directed edge leaving $a$ and pointing to $\pi(a)$.
\end{definition}

$G(A,\pi)$ possesses a special structure of not containing any cycles of length greater than $2$.

\begin{lemma} \label{lemma:cycles}
	The graph $G(A,\pi)$ contains cycles only of length $2$.
\end{lemma}
\begin{proof}
	Let $\{a_1, \ldots, a_s\}$ be a cycle such that for each $1 \le i \le s$ we have $\pi(a_i) = a_{i+1}$ (additions should be interpreted modulo $s$).
	We will show that $s=2$.
	
	By definition of $\pi$, it must be that $d(a_i, a_{i+1}) \le d(a_{i-1}, a_i)$ since $\pi(a_i) = a_{i+1}$.
	Then we have $d(a_1,a_2) \le \ldots \le d(a_s,a_1) \le d(a_1,a_2)$ and the chain of inequalities implies equality.
	Let $a_t$ be the element of the cycle that is greatest according to the ordering of $A$.
	Since $a_t$ and $a_{t+2}$ are equidistant from $a_{t+1}$, the criterion for breaking ties in Definition~\ref{definition:pi} ensures that $\pi(a_{t+1}) = a_t$.
	Since $\pi(a_{t+1}) = a_{t+2}$, it must be that $t+2 = t$ and so $s=2$.
\end{proof}
In light of Lemma~\ref{lemma:cycles}, let us consider the structure of the directed graph $(A,\pi)$.
Removing the edges in length-2 cycles, we are left with a forest (a collection of trees directed to the root).
Considering the full graph along with these 2-cycles, we see that each component is a pair of trees whose roots are coupled.
This forest of ``bi-trees'' can be 2-colored, and the following lemma shows that we can do this efficiently.

\begin{lemma} \label{lemma:2color}
	Given the function $\pi$, the graph $G(A,\pi)$ can be 2-colored in $O(m)$ time.
\end{lemma}
\begin{proof}
	We say that $\pi(a)$ is the parent of $a$, and that $a$ is the child of $\pi(a)$.
	Each vertex has exactly one parent, and the edges point to the parent.
	For each point $a \in A$, we store a pointer to its parent as well as a list of pointers to its children.
	Given the function $\pi$, this can be accomplished in $O(n)$ time.
	
	As reasoned above, the graph can be partitioned into bi-tree components.  We use the following iterative procedure until all vertices have been colored:
	(1) Select any uncolored vertex;
	(2) Walk along the edges until reaching the two roots;
	(3) Color each root a different color;
	(4) Recursively color each child vertex the opposite color than its parent.
	
	For Step 2, we will know we have located the roots when we return to the vertex we just left.
	This process (of moving from $a$ to $\pi(a)$ until reaching the root) terminates in time proportional to the depth of the tree.
	Therefore a total of $O(n)$ time is spent during Step 2 over all iterations of this procedure.
	
	For Step 4, finding a child takes $O(1)$ time since we have stored a list of children with each vertex.
	Therefore a total of $O(n)$ time is spent during Step 4 over all iterations of this procedure.
\end{proof}

We will need the following technical lemma to bound $\cost(A,Z)$.
Recall that $\opt_k(A)$ is defined using centers from anywhere in the metric space $\mathcal{X}$.
The lemma says that if we restrict to centers from $A$ itself, 
then the optimal cost increases by less than a factor of two. 


\begin{lemma} \label{lemma:2opt}
Let $\overline{\opt}_k(A)$ denote the minimum of $\cost(A,B)$ where $B \subset A$ ranges over all sets of $k$ points.  Then $\overline{\opt}_k(A) < 2 \opt_k(A)$.
\end{lemma}
\begin{proof}
Let $C \subset \mathcal{X}$ be a set of $k$ points such that $\cost(A,C) = \opt_k(A)$.
For each $c \in C$, let $c'$ be the closest point of $A$ to $c$.
Any element $a \in A$ that was connected to $c$ can be instead connected to $c'$ with a cost of $d(a,c') \le d(a,c) + d(c,c') \le 2 d(a,c)$ since $d(c,c') \le d(c,a)$.
Moreover, the cost of this cluster increased by strictly less than factor of two since if $c \notin A$ the cost decreased for $a = c'$ and if $c \in A$ the cost stayed the same.
Define $C' = \{c'\}_{c \in C}$.
Then $C' \subset A$ is a set of $k$ points such that $\overline{\opt}_k(A) \le \cost(A,C') < 2 \cost(A,C) = 2\opt_k(A)$.
\end{proof}

The $\mathcal{B}(X,k)$ algorithm is presented in the next theorem.
The basic idea is to $2$-color $G(A,\pi)$ and then eliminate one of the colors by relocating those points to their image under $\pi$.
Since $\pi$ maps each point to a point of the opposite color, this transformation increases the weights of one color while completely eliminating the other.

\begin{SCfigure}
	\centering
	\includegraphics[scale=0.5]{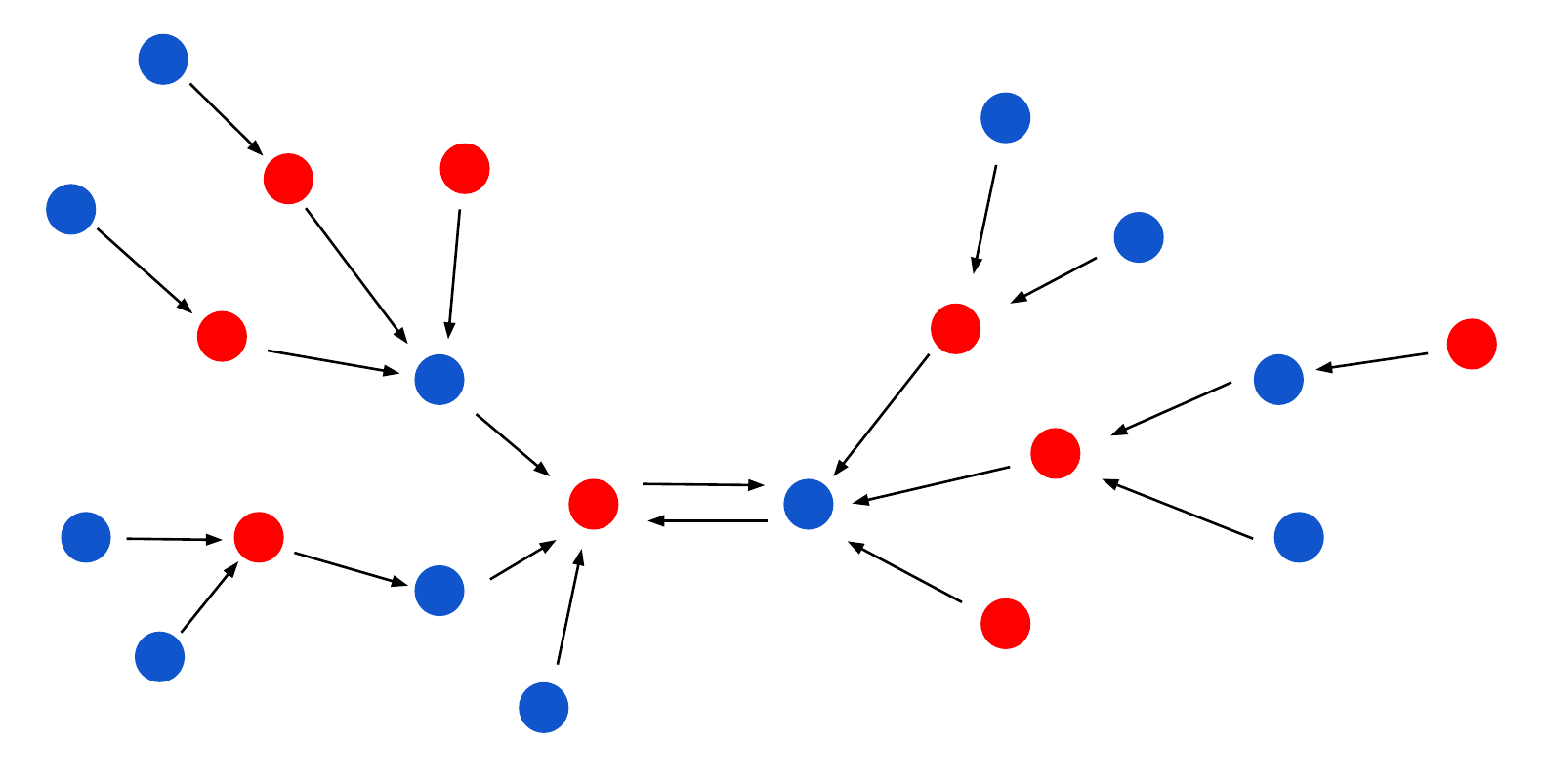}
\caption{A $2$-colored bi-tree of the nearest neighbor graph $G(A,\pi)$.
Algorithm $\mathcal{B}(A,k)$ removes all blue points by moving the weight of each blue point to the red point to which it points.}
\end{SCfigure}

\begin{theorem} \label{theorem:compress}
	Let $A$ be a weighted set of $n$ distinct points.  Assume that $\pi$ has been computed for $A$.  In $O(n)$ time, we can compute a weighted set $Z$ of at most $\lfloor \frac{n+k}{2} \rfloor$ distinct points such that $\cost(A,Z) < 2 \opt_k(A)$.
\end{theorem}
\begin{proof}
Let the function $w : A \rightarrow \mathbb{N}$ map each point of $A$ to its weight.
By Lemma~\ref{lemma:2color}, we $2$-color $A$ in $O(n)$ time.
Let $A_1$ and $A_2$ be the partition of $A$ into the two colors after removing the $k$ points with the top values of $w(a) d(a, \pi(a))$.
Let $|A_t|$ be the larger component (by number of points) and note that $|A_t| \ge \lceil \frac{n-k}{2} \rceil$.
	
Build $Z$ from $A$ as follows:
for each $a \in A_t$, increment $w(\pi(a))$ by $w(a)$ and delete $a$.
This procedure terminates in $O(n)$ time.
By definition of a 2-coloring, $\pi(a) \notin A_t$ for every element $a \in A_t$, and so $Z$ contains $n - |A_t| \le \lfloor \frac{n+k}{2} \rfloor$ weighted points.

Observe that $\cost(A,Z) = \sum_{a \in A_t} w(a) d(a, \pi(a))$.
The optimal $k$-median solution of $A$ using points from $A$ involves moving $n-k$ points of $A$ by at least the distance to the nearest neighbor, so $\sum_{a \in A_1 \cup A_2} w(a) d(a, \pi(a)) \le \overline{\opt}_k(A)$ and so $\cost(A,Z) \le \overline{\opt}_k(A)$.
This completes the proof since Lemma~\ref{lemma:2opt} guarantees $\overline{\opt}_k(A) < 2 \opt_k(A)$.
\end{proof}

To complete the guarantee of Theorem~\ref{thm:compress}, observe that we can return the exact value of $\cost(A,Z)$ in $O(n)$ time. 

\section{Discussion of Models for Semirandom Order} \label{section:open}

IITK Open Problem \#15 addresses computational models of semirandom-order for streams and asks how the models of $t$-bounded adversarial order and $\epsilon$-generated random order relate to each other~\cite{sublinear15}.
One can verify that $\epsilon = 0$ like $t = 0$ is equivalent to random-order and that $\epsilon = 1$ like $t = n$ is equivalent to adversarial-order.
However, we demonstrate in the following two lemmas that no other relations hold between these models.

\begin{lemma}
To simulate the $t$-bounded adversarial order model for any $t > 1$, the $\epsilon$-generated order model requires $\epsilon \ge 1-2^{-n/2}$.
\end{lemma}
\begin{proof}
Let the stream consist of elements $\{1, \ldots, n\}$.
Let $\chi(i)$ denote the identity of the $i^\text{th}$ element in the initial random-order stream.
For any $t \ge 2$, a $t$-bounded adversary can ensure that $\chi(2i-1) < \chi(2i)$ for all $1 \le i \le \frac{n}{2}$.
Let $A_i$ be the event that $\chi(2i-1) < \chi(2i)$.
In random-order, observe that $P(A_i) = \frac{1}{2}$ and that the $\{A_i\}_{i=1}^{n/2}$ are mutually independent.
Therefore the uniform distribution $\mu$ assigns probability mass $1 - (\frac{1}{2})^{n/2}$ to the orders satisfying $\chi(2i-1) > \chi(2i)$ for some $1 \le i \le \frac{n}{2}$.
Any distribution $\nu$ that assigns probability mass $1$ to $\cap_{i=1}^{n/2} A_i$ must satisfy $|| \mu - \nu ||_1 \ge 2(1-2^{-n/2})$
\end{proof}

\begin{lemma}
To simulate the $\epsilon$-generated order model for any $\epsilon \ge 2^{-n/10}$, the $t$-bounded adversarial order model requires $t > n/2$.
\end{lemma}
\begin{proof}
Let $X$ be a subset of $r$ elements.
In random-order, let $A$ be the event at at least one element of $X$ arrives among the first $\beta n$ elements.
$P_\mu(A) \ge 1 - (1-\beta)^r$, so we can create a distribution $\nu$ such that $P_\nu(A) = 1$ and $||\mu - \nu||_1 \le 2(1-\beta)^r$.

In the $t$-bounded adversarial order model, elements are \textit{sent} in random-order but intercepted by an adversary who manipulates the order that elements \textit{arrive} for the algorithm.
Observe that an element among the last $\frac{n-t}{2}$ to be sent cannot be among the first $\frac{n-t}{2}$ to arrive.
If $r \le \frac{n-t}{2}$ then with positive probability all elements of $X$ are among the last $\frac{n-t}{2}$ to be sent.
Therefore if $\beta n \le \frac{n-t}{2}$ we have $P(A) < 1$.
Setting $r = \beta n$, this necessitates that $t > (1 - 2\beta)n$ whenever $\epsilon \ge (1-\beta)^{\beta n}$.
The result follows by setting $\beta = \frac{1}{4}$.
\end{proof} 

Our matching bounds for online facility location show a non-trivial degradation of performance in the $t$-bounded adversarial-order model that smooothly interpolates between random-order and adversarial-order.
This supports the claim that $t$-bounded adversarial-order is a viable model of semi-randomness.
In contrast, it is trivial to show matching bounds of $\Theta(1 +  \frac{\epsilon \log n}{\log \log n})$ in the $\epsilon$-generated model.
More generally, any bound in expectation which is $\Theta(f)$ in random-order and $\Theta(g)$ in adversarial-order implies a $\Theta(f + \epsilon g)$ bound in the $\epsilon$-generated model.
As a result, the $\epsilon$-generated model is not interesting for a wide class of problems.
This class is rather large since any $\Theta(f \log \frac{1}{\delta})$ bound with probability $1-\delta$ implies a $\Theta(f)$ bound in expectation.

We conclude with an open question:
\paragraph{Open Question:} 
Given an adversarial-order bound of $O(f(n))$, all the results in this paper present a bound of $O(f(t))$ for $t$-bounded adversarial order, showing a smooth degradation as $t$ increases.
However, some problems exhibit a sharp phase transition.
For example, the size of the largest component in an Erd\H{o}s-R\'{e}nyi graph $ER(n,p)$ jumps from $O(\log n)$ to $\Omega(n)$ around $p \sim 1/n$.
In the $t$-bounded adversarial order model, is it always the case that bounds degrade smoothly as $t$ increases?  Alternatively, do problems exist that exhibit a sharp jump in some quantity of interest (i.e. time, space, or approximation factor) when $t$ increases by only a constant factor around some value?


\bibliographystyle{plain}
\bibliography{bibliography}

\end{document}